\documentclass[11pt]{amsart}

\def\cNB{{\mathcal {NB}}}

\def\cC{{\mathcal C}}

\def\cN{{\mathcal N}}

\def\Z{{\mathbb Z}}

\def\Z{{\mathbb Z}}

\def\aa{{\bf a}}

\def\uu{{\bf u}}
\def\vv{{\bf v}}
\def\ww{{\bf w}}
\def\xx{{\bf x}}
\def\yy{{\bf y}}
\def\zz{{\bf z}}

\def\00{{\bf 0}}
\def\+{\oplus}

\def\\{\cr}
\def\({\left(}
\def\){\right)}

\newcommand{\BBZ}{\mathbb{Z}}
\newcommand{\BBR}{\mathbb{R}}

\usepackage{amsfonts}
\usepackage{amsmath}
\usepackage{mathrsfs}
\usepackage{amssymb}
\usepackage{lscape}
\usepackage{graphicx}
\usepackage{array}
\usepackage{multirow}
\usepackage{comment}
\usepackage{xcolor}
\fontsize{12}{17}
\usepackage{enumerate}
\newcolumntype{M}[1]{>{\centering\arraybackslash}m{#1}}
\usepackage[left=2.4cm, right=2.4cm, bottom=3.30cm]{geometry}
\newtheoremstyle{case}{}{}{}{}{}{:}{ }{}

\numberwithin{subcase}{case}

\numberwithin{subsubcase}{subcase}
\newtheorem{theorem}{Theorem}[section]
\newtheorem{lemma}[theorem]{Lemma}

\newtheorem{corollary}[theorem]{Corollary}

\newtheorem{example}[theorem]{Example}
\newtheorem{remark}[theorem]{Remark}

\title{A new class of negabent functions
}

\author{Deep Singh and Maheshanand Bhaintwal}

\address{Deep Singh \\
Department of Mathematics and Statistics \\
Central University of Punjab\\
 Bathinda, India-151401}
\email{deepsingh.com}

\address{Maheshanand Bhaintwal \\
Department of Mathematics \\
Indian Institute of Technology Roorkee, India-247667}
\email{maheshanand@ma.iitr.ac.in}

\begin{document}
\maketitle

\begin{abstract}
	Negabent functions were introduced as a generalization of bent functions, which have applications in coding theory and cryptography.  In this paper, we have extended the notion of negabent functions to the functions defined from $\mathbb{Z}_q^n$ to $\mathbb{Z}_{2q}$ ($2q$-negabent), where $q \geq 2$ is a positive integer and $\mathbb{Z}_q$  is the ring of integers modulo $q$. For this, a new unitary transform (the nega-Hadamard transform) is introduced in the current set up, and some of its properties are discussed. Some results related to $2q$-negabent functions are presented. We present two constructions of $2q$-negabent functions. In the first construction, $2q$-negabent functions on $n$ variables are constructed when $q$ is an even positive integer. In the second construction, $2q$-negabent functions on two variables are constructed for arbitrary positive integer $q \ge 2$. Some examples of $2q$-negabent functions for different values of $q$ and $n$ are also presented.
\end{abstract}



\section{Introduction}
Riera and Parker  \cite{P2000,RP06gen} have extended the concept of  bent functions \cite{R76} to some generalized bent criteria where  Boolean functions are
required to have flat spectra with respect to one or more unitary transformations. The transforms they have
chosen are $n$-fold tensor products of the identity matrix
$\left(
    \begin{array}{cc}
      1 & 0 \\
      0 & 1 \\
    \end{array}
  \right),$
the Walsh-Hadamard matrix  $\frac{1}{\sqrt{2}} \left(
             \begin{array}{cc}
               1 & 1 \\
               1 & -1 \\
             \end{array}
           \right),$
and the nega-Hadamard matrix $ \frac{1}{\sqrt{2}} \left(
                                                    \begin{array}{cc}
                                                      1 & \imath \\
                                                      1 & -\imath \\
                                                    \end{array}
                                                  \right)$, with $\imath^2=-1.$
The \emph{nega-Hadamard transform} (NHT)  of a Boolean function $f$ at $\uu
\in \BBZ_2^n$ is a complex-valued function on $\mathbb{Z}_2^n$
defined as $$H_f (\uu) = \frac{1}{2^\frac{n}{2}}\sum_{\xx \in
\mathbb{Z}_2^n} {(-1)}^{f(\xx)+ \langle\xx, \uu\rangle}
\imath^{w_H(\xx)}~,$$ where $\langle\xx, \uu\rangle$ denotes the usual inner product of $\xx$ and $\uu$.  A Boolean function $f$ is said to be
\emph{negabent} if $\left|H_f(\uu)\right|= 1 $ for all $ \uu \in
\mathbb{Z}_2^n$. The multiset $\{H_f (\uu): ~\uu \in \BBZ_2^n \}$ is
called the \emph{nega-Hadamard spectrum} of $f$.
The sum
$$ C_{f, g} (\uu) = \sum_{\xx \in
\mathbb{Z}_2^n} {(-1)}^{f(\xx)+ g(\xx + \uu)} (-1)^{\langle\xx,
\uu\rangle}$$ is called the \emph{nega-crosscorrelation} between the
Boolean functions $f$ and $g$ at $ \uu \in \mathbb{Z}_2^n$. For
$f=g$, the nega-crosscorrelation
\[C_{f, f}= C_{f}= \sum_{\xx \in \mathbb{Z}_2^n} {(-1)}^{f(\xx)+
f(\xx + \uu)} (-1)^{\langle\xx, \uu\rangle}\]
 is called the
\emph{nega-autocorrelation} of $f$ at $\uu \in \mathbb{Z}_2^n$.

In recent years, the construction of Boolean negabent functions has emerged as an important problem.
The authors in \cite{P2000,PP07,RP06gen,SPP08,SGCGM12,SPT12,ZWP15} have presented several properties and constructions of Boolean negabent
functions.

 As a generalization of Boolean functions, $q$-ary functions have been studied in \cite{KSW85,SBS13}. Recently, Schmidt \cite{S15} has  given a construction of  bent functions for the functions from $\BBZ_q^n$ to $\BBZ_{2q}$.
In this paper,  we present a new class of negabent functions
by considering  functions from $\BBZ_q^n$ to $\BBZ_{2q}$.

Let $\cNB_{n, q}$ be the set of all
functions from $\BBZ_q^n$ to $\BBZ_{2q}$. Let $\xx=(x_1, \ldots, x_n )\in \BBZ_q^n$.  Then we define
$\hat{\xx}=(\hat{x}_1, \ldots, \hat{x}_n )\in \BBZ^n$ where $
\hat{x}_i= \left \{
\begin{array}{ll}
                     x_i,  &  \mbox{  if  }  x_i \geq 0, \\
                       x_i + q,  & \mbox{ if }  x_i < 0.\\
       \end{array}\right.$
Thus $\hat{x}_i$ is the least non-negative integer such that $\hat{x}_i = x_i$ modulo $q$.
Let $f \in \cNB_{n, q}$, $\xi$ a primitive $q^{th}$ root of unity, and $\omega$ be a primitive $2q^{th}$ root
of unity. Then we define the \emph{$2q$-nega-Hadamard transform} ($2q$-NHT) $\cN_f(\uu)$ of $f$ at $\uu\in\BBZ_q^n$ as
$$\cN_f (\uu) =
\frac{1}{q^\frac{n}{2}}\sum_{\xx \in \mathbb{Z}_q^n}
{\omega}^{f(\xx)} \xi^{\langle \hat{\xx}, \hat{\uu}\rangle}
\omega^{\Sigma \hat{x}_i}.$$
 A function $f \in \cNB_{n, q}$ is said to be a \emph{$2q$-negabent function} if $\left|\cN_f(\uu)\right|= 1 $ for every $ \uu \in
\mathbb{Z}_q^n$. The multiset $\{\cN_f (\uu): ~\uu \in \BBZ_q^n \}$ is called the  \emph{$2q$-nega-Hadamard spectrum} of $f$.

We remark here that our motivation for studying this class of negabent functions comes
mainly from the fact that if we consider $q$-ary negabent
functions from $\BBZ_q^n$ to $\BBZ_q$ with $q$-nega-Hadamard
transform
$$\cN_f^{\prime} (\uu) =
\frac{1}{q^\frac{n}{2}}\sum_{\xx \in \mathbb{Z}_q^n} {\xi}^{f(\xx)}
\xi^{ \langle \hat{\xx}, \hat{\uu}\rangle} \omega^{\Sigma \hat{x}_i},
$$
which is a natural generalization of nega-Hadamard transform of a Boolean function, then we are not able to find any function with flat
spectra for $q>2$, despite several attempts for the same through computer search. In fact, it is not known to
us whether such functions exist at all for $q > 2$. What we know is that
for $q>2$, affine functions are not negabent functions for this class (see Theorem \ref{thmoldnega} below).
On the other
hand, in the new setup we have proposed, we get many interesting
examples of $2q$-negabent functions for various values of $q$
and $n.$

Let $f, g \in \cNB_{n, q}$ and $\uu \in \BBZ_q^n$. Then
the sum
 $$ \cC_{f, g}^q(\uu)= \sum_{\xx
\in \mathbb{Z}_q^n} \omega^{f(\xx)-g(\xx + \uu)}
(-1)^{n_q(\hat{\xx}, \hat{\uu})}~,$$ where $n_q(\hat{\xx},
\hat{\uu})=\sum_{i=1}^n \left\lfloor \frac{\hat{x}_i+
\hat{u}_i}{q}\right\rfloor=|\{i : ~ \hat{x}_i + \hat{u}_i \geq q \}|$, is
called the \emph{$2q$-nega-crosscorrelation} ($2q$-NCC) between $f$ and $g$ at $\uu$.
 The identity $\sum_{i=1}^n \left\lfloor
\frac{\hat{x}_i+\hat{u}_i}{q}\right\rfloor=|\{i : ~\hat{x}_i + \hat{u}_i
\geq q \}|$ holds in the present case as $\hat{x}_i + \hat{u}_i <
2q$ for all $x_i, u_i \in \BBZ_q$. For $f =g$, the quantity
$$\cC_f^q(\uu)=\cC_{f, f}^q(\uu)= \sum_{\xx \in
\mathbb{Z}_q^n} \omega^{f(\xx)-f(\xx + \uu)} (-1)^{n_q(\hat{\xx},
\hat{\uu})} $$ is called  the \emph{$2q$-nega-autocorrelation} ($2q$-NAC) of $f$ at $\uu \in \mathbb{Z}_q^n$.
The following is an important basic result.

\begin{lemma}\cite{KSW85} \label{chap5-lemimp}
Let $n$ be a positive integer and $\uu
\in \Z_q^n$. Then
\begin{eqnarray*}\label{eq03}
\sum_{\xx \in \Z_q^n}\xi^{<\uu, ~\xx>}=
    \left \{ \begin{array}{ll}
                    q^{n},  &  \mbox{  if  }  \uu=0, \\
                       0,  & \mbox{ otherwise }.\\
       \end{array}  \right.
\end{eqnarray*}
\end{lemma}

In this paper, we investigate some properties of $2q$-NHT and $2q$-negabent functions.  We generalize a result of Schmidt et al. \cite[Lemma 1]{SPP08} (obtained for binary case) to the current setup. We present two constructions of $2q$-negabent functions. In the first construction, $2q$-negabent functions on $n$ variables are constructed when $q$ is an even positive integer. In the second construction, $2q$-negabent functions on two variables are constructed for arbitrary positive integer $q \ge 2$. Some examples of $2q$-negabent functions for different values of $q$ and $n$ have been given.


\section{Properties of $2q$-nega-Hadamard transform}

In this section, we present some properties of $2q$-NHT regarding its behavior on various
combinations of the functions in $\cNB_{n, q}$.

The following lemma is an important property and  will be used
frequently in this paper.
\begin{lemma}\label{lemsum} Let $\xx =(x_1, \ldots, x_n)$, $\yy = (y_1, \ldots, y_n)$, and $\zz= (z_1, \ldots, z_n)$ be in $\BBZ_q^n$ such that
$\zz=\xx+\yy.$ Then
\begin{equation*}
\sum \hat{z}_i=\sum \hat{x}_i + \sum \hat{y}_i- q
n_q(\hat{\xx}, \hat{\yy}).
\end{equation*}
\end{lemma}

\begin{proof} We have
\[
\hat{z}_i =\hat{x}_i + \hat{y}_i ~~\mbox{(mod q)}
=(\hat{x}_i + \hat{y}_i)-q \left\lfloor \frac{\hat{x}_i + \hat{y}_i}{q} \right\rfloor~.\]
So,
\[\sum \hat{z}_i = \sum \hat{x}_i + \sum \hat{y}_i - q \sum
\left\lfloor \frac{\hat{x}_i + \hat{y}_i}{q} \right\rfloor.
\]
Since $\hat{x}_i + \hat{y}_i < 2q,$ for $i=1, 2, \ldots, n$, we have
 $$\left\lfloor \frac{\hat{x}_i +
\hat{y}_i}{q} \right\rfloor =
    \left \{ \begin{array}{ll}
                     1,  &  \mbox{  if  }  \hat{x}_i + \hat{y}_i \geq q, \\
                       0,  & \mbox{ otherwise }.\\
       \end{array}  \right.  $$
Therefore, $\sum \left\lfloor \frac{\hat{x}_i + \hat{y}_i}{q} \right\rfloor = \left|\{i ~:~ \hat{x}_i + \hat{y}_i \geq q\}\right|= n_q(\hat{\xx},
\hat{\yy}).$
The result follows.
\end{proof}

The following theorem generalizes a result of Schmidt et al.
\cite[Lemma 1]{SPP08} (obtained for the binary case) to the current setup.
\begin{theorem}\label{lemimp2} For any $\uu\in \BBZ_q^n,$ we have
\begin{equation*}
\sum_{\xx \in\BBZ_q^n}\xi^{\langle \hat{\uu},
~\hat{\xx}\rangle}\omega^{\Sigma \hat{x}_j}=\frac{1}{\prod_{j=1}^n
\sin (2 \hat{u}_j + 1)\frac{\pi}{2q}} \eta^{n(q-1)-2 \Sigma
\hat{u}_j},
\end{equation*} where $\xi$ is a $q$-th, $\omega$ is a $2q$-th and
$\eta$ is a $4q$-th primitive complex root of unity.
\end{theorem}

\begin{proof} We have
\begin{eqnarray}\label{eq1}
\sum_{\xx \in\BBZ_q^n}\xi^{\langle \hat{\uu},
~\hat{\xx}\rangle}\omega^{\Sigma \hat{x}_j} &=&
\sum_{\xx \in\BBZ_q^n}\xi^{\Sigma \hat{u}_j \hat{x}_j}\omega^{\Sigma \hat{x}_j} \nonumber\\
&=& \prod_{j=1}^n \sum_{x_j \in \Z_q}\xi^{\hat{u}_j \hat{x}_j}\omega^{\hat{x}_j}\nonumber \\
&=& \prod_{j=1}^n \frac{1-(\omega \xi^{\hat{u}_j})^q}{1-\omega \xi^{\hat{u}_j}}\nonumber \\
&=& \prod_{j=1}^n \frac{2}{1-\omega \xi^{\hat{u}_j}}.
\end{eqnarray}
Since $\xi$ and $\omega$ are $q$-th and $2q$-th roots of unity respectively, we have
\begin{multline*}\label{eq2}
1-\omega \xi^{\hat{u}_j}
= 1-
e^{(2 \hat{u}_j +1)\frac{\pi \imath}{q}}\nonumber\\
=1- \cos (2 \hat{u}_j +1)\frac{\pi}{q}- \imath \sin (2 \hat{u}_j +1)\frac{\pi}{q}\nonumber\\
= 2\sin^2(2 \hat{u}_j +1)\frac{\pi}{2q}- 2\imath \sin (2 \hat{u}_j +1)\frac{\pi}{2q}\cos(2 \hat{u}_j +1)\frac{\pi}{2q}\nonumber\\
= 2\sin(2 \hat{u}_j +1)\frac{\pi}{2q}\left[ \sin (2 \hat{u}_j
+1)\frac{\pi}{2q}- \imath \cos (2 \hat{u}_j +1)\frac{\pi}{2q}\right]\nonumber\\
= 2e^{-(q-1-2\hat{u}_j)\frac{\pi \imath}{2q}}\sin(2 \hat{u}_j +1)\frac{\pi}{2q}~.
\end{multline*}
Then it follows that 
\begin{eqnarray*}\label{eq3}
\frac{2}{1-\omega \xi^{\hat{u}_j}} &=& \frac{e^{(q-1-2\hat{u}_j)\frac{\pi \imath}{2q}}}{\sin(2 \hat{u}_j +1)\frac{\pi}{2q}} \\
&=& \frac{\eta^{(q-1-2\hat{u}_j)}}{\sin(2 \hat{u}_j +1)\frac{\pi}{2q}}~.
\end{eqnarray*}
Therefore, we have
\begin{eqnarray*}
\sum_{\xx \in\BBZ_q^n}\xi^{\langle\hat{\uu},
~\hat{\xx}\rangle}\omega^{\Sigma \hat{x}_j} &=& \prod_{j=1}^n \frac{2}{1-\omega \xi^{\hat{u}_j}} \\
&=& \prod_{j=1}^n \frac{\eta^{(q-1-2\hat{u}_j)}}{\sin(2 \hat{u}_j +1)\frac{\pi}{2q}}\\
&=& \frac{\eta^{n(q-1)-2 \Sigma \hat{u}_j}}{\prod_{j=1}^n \sin (2 \hat{u}_j +1)\frac{\pi}{2q}}
~.
\end{eqnarray*}
Hence the result.
\end{proof}

\begin{remark}
In the binary case, i.e. for $q=2$, $\prod_{j=1}^n \sin (2 \hat{u}_j
+1)\frac{\pi}{2q}=\prod_{j=1}^n \sin (2 \hat{u}_j
+1)\frac{\pi}{4}=\frac{1}{2^{n/2}}, $ as $\sin (2 \hat{u}_j
+1)\frac{\pi}{4}=\frac{1}{\sqrt{2}}$ for ${u}_j\in\{0, 1\}.$
However, for $q>2$ the product $\prod_{j=1}^n \sin (2 \hat{u}_j
+1)\frac{\pi}{2q}$ depends on the values of ${u}_j$ and is in general
not a constant.
\end{remark}

\begin{theorem}\label{thmoldnega}
For $q>2$, there is no affine function $f:\mathbb{Z}_q^n \rightarrow \mathbb{Z}_q$ such that $|\mathcal{N}^\prime(\uu)| = 1$ for all $\uu \in \mathbb{Z}_q^n$.
\end{theorem}

\begin{proof}
Let $q>2$, and let $f:\mathbb{Z}_q^n \rightarrow \mathbb{Z}_q$  be an affine function, given by
\[f(\xx) = \langle \aa, \xx\rangle + b, \quad \xx, ~ \aa \in \mathbb{Z}_q^n, ~ b \in \mathbb{Z}_q~.\]
Then for any $\uu \in \mathbb{Z}_q^n$
\begin{align*}
\mathcal{N}^\prime(\uu) & = \frac{1}{q^\frac{n}{2}}\sum_{\xx \in \mathbb{Z}_q^n} {\xi}^{\langle \aa, \xx \rangle+b}
\xi^{\langle\hat{\xx}, \hat{\uu}\rangle} \omega^{\Sigma \hat{x}_i}\\
 &= \frac{1}{q^\frac{n}{2}}\sum_{\xx \in \mathbb{Z}_q^n} {\xi}^{\langle \hat{\aa}, \hat{\xx} \rangle+b}
\xi^{\langle\hat{\xx}, \hat{\uu}\rangle} \omega^{\Sigma \hat{x}_i} \\
&=\frac{1}{q^\frac{n}{2}}\xi^b\sum_{\xx \in \mathbb{Z}_q^n} {\xi}^{\langle \hat{\aa}+\hat{\uu}, \hat{\xx} \rangle} \omega^{\Sigma \hat{x}_i} \\
&=\frac{1}{q^\frac{n}{2}}\xi^b\sum_{\xx \in \mathbb{Z}_q^n} {\xi}^{\langle \hat{\zz}, \hat{\xx} \rangle} \omega^{\Sigma \hat{x}_i}~,
\end{align*}
where $\zz = \aa+\uu$. Then by Theorem \ref{lemimp2},
\[\mathcal{N}^\prime(\uu) = \frac{1}{q^\frac{n}{2}}\xi^b\frac{1}{\prod_{j=1}^n \sin (2 \hat{z}_j +1)\frac{\pi}{2q}}
\eta^{n(q-1)-2 \Sigma \hat{z}_j}~,\]
where $\eta$ is a primitive $4q$-th root of unity. Then
\begin{equation}\label{eqoldnega}
\left|\mathcal{N}^\prime(\uu)\right| = \frac{1}{q^\frac{n}{2}}\frac{1}{\left|\prod_{j=1}^n \sin (2 \hat{z}_j +1)\frac{\pi}{2q}\right|}~.
\end{equation}
Now consider two vectors $\uu_1, \uu_2 \in \mathbb{Z}_q^n$ such that
\begin{align*}
\aa + \uu_1 & = \zz_1 = (0, z_2, z_3, \ldots, z_n)~, \\
\aa + \uu_2 & = \zz_2 = (1, z_2, z_3, \ldots, z_n)~,
\end{align*}
where $z_j \in \mathbb{Z}_q$ are arbitrary. Since, for any fixed $\aa \in \mathbb{Z}_q^n$, $\aa+\zz$ runs over $\mathbb{Z}_q^n$ as $\zz$ runs over $\mathbb{Z}_q^n$, two such elements $\uu_1, \uu_2 \in \mathbb{Z}_q^n$ can always be found.
Then from (\ref{eqoldnega}), we get
\[\frac{|\mathcal{N}^\prime(\uu_1)|}{\left|\mathcal{N}^\prime(\uu_2)\right|} = \frac{\left|\sin \frac{3\pi}{2q}\right|}{\left|\sin \frac{\pi}{2q}\right|}~.\]
Now for $q > 2$,
\[ \frac{\left|\sin \frac{3\pi}{2q}\right|}{\left|\sin \frac{\pi}{2q}\right|} = \frac{\sin \frac{3\pi}{2q}}{\sin \frac{\pi}{2q}} > 1~,\]
because $0 < \frac{\pi}{2q} < \frac{3\pi}{2q} \le \frac{\pi}{2}$, and $\sin \theta$ is strictly increasing in the interval $0\le \theta\le \frac{\pi}{2}$.
Therefore, $\left|\mathcal{N}^\prime(\uu_1)\right| > \left|\mathcal{N}^\prime(\uu_2)\right|$, and hence $\left|\mathcal{N}^\prime(\uu_1)\right|$ and $\left|\mathcal{N}^\prime(\uu_2)\right|$ cannot both be $1$. Hence the result.
\end{proof}

Like in the Boolean case \cite{PP07}, the $2q$-NHT is also a unitary transformation. The
following result gives the inverse of $2q$-NHT of a function
$f\in\cNB_{n, q}$.
\begin{lemma} Let $f\in\cNB_{n, q}$. Then for all $\xx\in \BBZ_q^n,$
we have
\begin{equation*}
\omega^{f(\xx)}= q^{\frac{-n}{2}} \omega^{-\Sigma \hat{x}_i}
\sum_{\uu\in \BBZ_q^n} \cN_f(\uu) \xi^{\langle-\hat{\uu},
\hat{\xx}\rangle}.
\end{equation*}
\end{lemma}
\begin{proof}The result follows from definition of $2q$-NHT and Lemma \ref{chap5-lemimp}.
\end{proof}

In the next result, we show that the conservation law for the $2q$-NHT
values of $f \in \cNB_{n, q}$ holds. We call it the
\emph{$2q$-nega-Parseval's identity}.

\begin{theorem}
 Let $f\in \cNB_{n, q}.$ Then $$\displaystyle \sum_{\uu\in \BBZ_q^n}
 \left|\cN_f(\uu)\right|^2 = q^{n}.$$
\end{theorem}

\begin{proof} The proof follows directly from definition of $2q$-NHT.
\end{proof}

The next result gives a relationship between $2q$-NHT of $f, g\in \cNB_{n, q}$ and
their $2q$-nega-crosscorrelation.

\begin{theorem}\label{thmcwfwg} If $f, g \in \cNB_{n, q}$ and $\uu, \zz \in
\mathbb{Z}_q^n$, then
\begin{equation*} 
\begin{split}
\sum_{\zz \in \mathbb{Z}_q^n} \cC_{f, g}^q(\zz)\omega^{- \Sigma
\hat{z}_i}\xi^{\langle- \hat{\uu}, \hat{\zz}\rangle}
=q^n \cN_f(\uu)\overline{\cN_g(\uu)}, \mbox{ and } \\
\cC_{f, g}^q(\zz)= \omega^{\Sigma \hat{z}_i} \sum_{\uu \in
\mathbb{Z}_q^n} \cN_f(\uu)\overline{\cN_g(\uu)} ~~\xi^{\langle
\hat{\uu},~ \hat{\zz} \rangle}.
\end{split}
\end{equation*}
\end{theorem}

\begin{proof} By the definition of  $2q$-NHT, we have
\begin{align}  \nonumber
&\cN_f (\uu) \overline{\cN_g (\uu)} \\  \nonumber
&= \frac{1}{q^n}\sum_{\xx \in
\mathbb{Z}_q^n} \omega^{f(\xx)} {\xi}^{\langle \hat{\uu}, \hat{\xx}
\rangle} \omega^{\Sigma \hat{x}_i} \sum_{\yy \in
\mathbb{Z}_q^n}\omega^{-g(\yy)} \xi^{- \langle \hat{\uu}, \hat{\yy}
\rangle}
\omega^{-\Sigma \hat{y}_i} \\ \nonumber
&= \frac{1}{q^n}\sum_{\xx, \zz \in \mathbb{Z}_q^n} {\omega}^{f(\xx)-
g(\xx + \zz)}\xi^{\langle \hat{\uu}, \hat{\xx}\rangle- \langle
\hat{\uu}, \hat{\xx} + \hat{\zz} \rangle} \omega^{-\Sigma \hat{z}_i
+q n_q(\hat{\xx}, \hat{\zz})} \\ \label{eq11}
&= \frac{1}{q^n}\sum_{\zz \in \mathbb{Z}_q^n} \omega^{-\Sigma
\hat{z}_i} \cC_{f, g}^q(\zz)\xi^{\langle-\hat{\uu},
\hat{\zz}\rangle}.
\end{align}
Then we get
\begin{align*}
&\omega^{\Sigma \hat{z}_i} \sum_{\uu \in \mathbb{Z}_q^n}\cN_f (\uu)
\overline{\cN_g (\uu)}\xi^{\langle \hat{\uu}, \hat{\zz}\rangle} \\
 &=
\frac{1}{q^n}\omega^{\Sigma \hat{z}_i} \sum_{\uu \in
\mathbb{Z}_q^n}\sum_{\xx \in \mathbb{Z}_q^n}\omega^{-\Sigma
\hat{x}_i}\cC_{f, g}^q(\xx)
\xi^{\langle \hat{\uu}, -\hat{\xx} + \hat{\zz}\rangle}\nonumber \\
=& \cC_{f, g}^q(\zz).
\end{align*}
\end{proof}

In particular, if $f=g$, then we have
\begin{equation}\label{eq123}  \cC_{f}^q(\zz) = \cC_{f, f}^q(\zz) = \omega^{\Sigma \hat{z}_i}
 \sum_{\uu\in
\mathbb{Z}_q^n} \left|\cN_f(\uu)\right|^2 \xi^{\langle \hat{\uu},
~\hat{\zz}\rangle}.
\end{equation}
\noindent By putting $\zz =\00$ in $(\ref{eq123})$, we obtain
$\sum_{\uu \in \mathbb{Z}_q^n}\left|\cN_f(\uu)\right|^2 = q^{n},$ which is the
$2q$-nega-Parseval's identity.

\begin{corollary}\label{coronegabent}
A function $f \in \mathcal \cNB_{n, q}$ is $2q$-negabent if
and only if $\cC_f^q(\uu)=0$ for all $\uu \in
\BBZ_q^n\setminus\{0\}$.
\end{corollary}

\begin{proof} The proof follows from
Lemma \ref{chap5-lemimp} and $(\ref{eq123})$.
\end{proof}

\section{Characterization of $2q$-negabent functions}\label{sec4}

Recall that for any fixed $\vv = (v_1, \ldots, v_r)$ with $1 \leq r
\leq n$ and $f \in \mathcal \cNB_{n, q}$, the restriction $f_{\vv}$
of $f$ is
$$f_{\vv}(x_1, \ldots, x_{n-r}) =
f(v_1, \ldots, v_r, x_1, \ldots, x_{n-r}).$$
 Also, let $\uu\ww$ denote the concatenation $(u_1, \ldots, u_r, w_1, \ldots$, $w_{n-r})$ of any two  vectors $\uu = (u_1,
\ldots, u_r) \in \BBZ_q^r$ and $\ww = (w_1, \ldots, w_{n-r}) \in
\BBZ_q^{n-r}$.

\begin{lemma}\label{lemnegaconcatenation}
Let $\uu \in \BBZ_q^r$, $\ww \in \BBZ_q^{n-r}$ and $f \in \cNB_{n,
q}.$ Then the $2q$-nega-autocorrelation ($2q$-NAC) of $f$ is given by
\begin{equation*}
\cC_f^q(\uu \ww) = \sum_{\vv \in \BBZ_q^r}\cC_{f_{\vv}, f_{\vv +
\uu}}^q(\ww) (-1)^{n_q(\hat{\uu}, \hat{\vv})}.
\end{equation*}
\end{lemma}

\begin{proof} We have
\begin{equation*}
\begin{split}
\cC_f^q(\uu\ww) &= \sum_{\xx \in \BBZ_q^n }\omega^{f(\xx) - f(\xx
+\uu\ww)}(-1)^{n_q(\hat{\uu} \hat{\ww}, \hat{\xx})} \\
&=\sum_{\vv \in \BBZ_q^r }\sum_{\zz \in \BBZ_q^{n-r}
}\omega^{f(\vv\zz) - f(\vv\zz +\uu\ww)}(-1)^{n_q(\hat{\uu}
\hat{\ww}, \hat{\vv} \hat{\zz})},
\end{split}
\end{equation*}
\noindent where $\xx \in \BBZ_q^n$ can be considered as a vector
concatenation of  $\vv \in \BBZ_q^r$ and  $\zz \in \BBZ_q^{n-r}$.
Also, for any fixed vectors $\vv$ and $\uu$, we have $f(\vv \zz)=f_{\vv}(\zz)$ and
$f(\vv\zz +\uu\ww)=f_{\vv+\uu}(\zz +\ww)$. Then
\begin{equation*}
\begin{split}
&\cC_f^q(\uu\ww) \\
& =\sum_{\vv \in \BBZ_q^r }\sum_{\zz \in \BBZ_q^{n-r}
}\omega^{f_{\vv}(\zz) - f_{\vv+\uu}(\zz +\ww)} (-1)^{n_q(\hat{\uu},
\hat{\vv})+
n_q(\hat{\ww}, \hat{\zz})}\\
&=\sum_{\vv \in \BBZ_q^r } \cC_{f_{\vv},
f_{\vv+\uu}}^q(\ww)(-1)^{n_q(\hat{\uu}, \hat{\vv})}. \hspace{5.0cm}
\end{split}
\end{equation*}
\end{proof}

Two functions $ f, g \in \cNB_{n, q}$ are said to have {\em
complementary $2q$-NAC} if for all $\uu \in
\BBZ_q^n\setminus \{\00\}$,  $\cC_f^q(\uu) + \cC_g^q(\uu) = 0 $.
In the following theorem, we present a relationship between the
$2q$-NHT of $f,g \in \BBZ_q^n$ and their $2q$-NAC.

\begin{theorem}\label{thmnegacomplimentaryauto}
 Two functions $ f, g \in \cNB_{n, q}$ have
 complementary $2q$-NAC if and only if
\begin{equation*}
|\cN_f(\uu)|^2 + |\cN_g(\uu)|^2 = 2 ~\mbox{for all} ~ \uu \in
\BBZ_q^n.
\end{equation*}
\end{theorem}

\begin{proof}
The proof follows from \eqref{eq11} and the definition of $2q$-NAC.
\end{proof}

\begin{remark}
For arbitrary positive
integers $r, s$ and $q$, the direct sum $f\in \cNB_{r+s, q}$  of two $2q$-negabent
functions $f_1 \in \cNB_{r, q}$ and $f_2 \in \cNB_{s, q}$ is $2q$-negabent.
\end{remark}

The next result presents a characterization
of $2q$-negabent functions for $q=4$ on $n+1$ variables in
terms of its subfunctions on $n$ variables.

\begin{theorem} A function $h \in \cNB_{n + 1, 4}$ expressed as
\begin{eqnarray*}
h(x_1, \ldots ,x_{n+1})=(1+ \hat{x}_{n+1})f(x_1,\hdots ,x_n) \\
+ \hat{x}_{n + 1}g(x_1, \ldots,x_n),
\end{eqnarray*}
where $f, g \in \cNB_{n, 4},$ is $2q$-negabent if and only if
\begin{enumerate}
\item [{(i)}] $|\sum_{j=0}^3 \omega^j \cN_{h_j}(\uu)|=2$ for all $\uu\in
\BBZ_4^n,$ where $\omega=(1+\imath)/\sqrt{2}$ is a primitive $8$-th
root
of unity.
\item [{(ii)}]
$\frac{\cN_{h_0}(\uu)-\omega^2\cN_{h_2}(\uu)}{ \omega
\cN_{h_1}(\uu)- \omega^3 \cN_{h_3}(\uu)}= \phi(\uu)$ and
$\frac{\cN_{h_0}(\uu)+ \omega^2\cN_{h_2}(\uu)}{\omega
\cN_{h_1}(\uu)+\omega^3\cN_{h_3}(\uu)} = \imath \psi(\uu)$,
$\phi(\uu),
\psi(\uu) \in \BBR$. 
\item [{(iii)}]
$\sum_{j=0}^3 \left|\cN_{h_j}(\uu)\right|^2=4$ for all $\uu\in \BBZ_4^n$, and
\vspace{.2cm} $\overline{\cN_{h_0}(\uu)}\cN_{h_2}(\uu)
 -\cN_{h_0}(\uu)\overline{\cN_{h_2}(\uu)}+\overline{\cN_{h_1}(\uu)}\cN_{h_3}(\uu)-\cN_{h_1}(\uu)\overline{\cN_{h_3}(\uu)}=0.$
\end{enumerate}
\end{theorem}

\begin{proof} The proofs of these assertions follow from the definition of $2q$-NHT.\end{proof}

\section{Constructions}

An interesting problem in cryptography is to construct functions  having flat spectra with respect to the transformation employed. In this section, we present two constructions of $2q$-negabent functions. In the first construction, $2q$-negabent functions are constructed when $q$ is an even positive integer. We first establish the following lemma.

\begin{lemma} \label{lemma5}
If $q$ is an even positive integer, then for any $a \in \mathbb{Z}_q$,
\[\sum_{x \in \mathbb{Z}_q} \omega^{x^2} = \sum_{x \in \mathbb{Z}_q} \omega^{(x+a)^2}~,\]
where $x^2$ and $(x+a)^2$ are computed modulo $2q$.
\end{lemma}

\begin{proof}
Since $q$ is even, $q^2 = 0$ (mod $2q$), and hence $(q+a)^2 = a^2$ (mod $2q$) for any $a \in \mathbb{Z}_q$. Therefore, the list of the squares of the elements of $\mathbb{Z}_{2q}$ is
\[0, 1^2, 2^2, \ldots, (q-1)^2, 0, 1^2, 2^2, \ldots, (q-1)^2~.\]
This implies that, for any $a \in \mathbb{Z}_q$, the list $0, 1^2, 2^2, \ldots, (q-1)^2$ (mod $2q$) and the list $a^2, (a+1)^2, (a+2)^2, \ldots, (a+q-1)^2$ (mod $2q$) are identical. Since $\omega$ is a $2q$th root of unity, it follows that
\[\sum_{x \in \mathbb{Z}_q} \omega^{x^2} = \sum_{x \in \mathbb{Z}_q} \omega^{(x+a)^2}~.\]
\end{proof}

\begin{theorem}\label{thm7}
Let $q$ be an even positive integer. Then the function $f: \BBZ_q^n \rightarrow \BBZ_{2q}$, defined by
$$f(x_1,  \ldots, x_n) = \hat{x}_1^2+ \cdots + \hat{x}_n^2 - \hat{x}_1 - \cdots - \hat{x}_n~,$$
is a $2q$-negabent function.
\end{theorem}

\begin{proof}
The $2q$-NHT of $f$ at $\uu \in \mathbb{Z}_q^n$ is
\begin{eqnarray*}
q^{n/2}\cN_f(\uu)  &=& \sum_{\xx \in \mathbb{Z}_q^n}
{\omega}^{\hat{x}_1^2+ \cdots \hat{x}_n^2 - \hat{x}_1 - \cdots -
\hat{x}_n}
\xi^{\langle\hat{\uu}, \hat{\xx}\rangle} \omega^{\Sigma \hat{x}_j}\\
&=& \sum_{\xx \in \mathbb{Z}_q^n} {\omega}^{\hat{x}_1^2+ \cdots
\hat{x}_n^2 +
2\hat{u}_1 \hat{x}_1 + \cdots + 2\hat{u}_n \hat{x}_n} \\
&=& \sum_{\xx \in \mathbb{Z}_q^n} \prod_{j=1}^n {\omega}^{\hat{x}_j^2 +
2\hat{u}_j \hat{x}_j}\\
&=& \prod_{j=1}^n \sum_{x_j \in \mathbb{Z}_q} {\omega}^{\hat{x}_j^2 +
2\hat{u}_j \hat{x}_j} \\
&=& \prod_{j=1}^n \left({\omega}^{ - \hat{u}_j^2} \sum_{x_j \in \mathbb{Z}_q} {\omega}^{(\hat{x}_j + \hat{u}_j)^2}\right)\\
&=& \prod_{j=1}^n {\omega}^{ - \hat{u}_j^2} \cdot \(\sum_{x_j \in
\mathbb{Z}_q}{\omega}^{\hat{x}_j^2}\)^n \\
&&(\mbox{using Lemma \ref{lemma5}}).
\end{eqnarray*}
Now we have
\begin{eqnarray} \nonumber
\begin{split}
\left|\sum_{x \in \mathbb{Z}_q} {\omega}^{\hat{x}^2}\right|^2 & =
\sum_{x \in \mathbb{Z}_q} {\omega}^{\hat{x}^2}
 \sum_{y \in
\mathbb{Z}_q} {\omega}^{-\hat{y}^2}\\
& = \sum_{x\in \mathbb{Z}_q} {\omega}^{\hat{x}^2} \sum_{t \in
\mathbb{Z}_q} {\omega}^{-(\hat{x}+\hat{t})^2} \\
&(\mbox{using Lemma \ref{lemma5}})\\
& =\sum_{t \in \mathbb{Z}_q} {\omega}^{\hat{t}^2} \sum_{x\in
\mathbb{Z}_q} {\xi}^{-\hat{x}\hat{t}} = q~,\\
\end{split}
\end{eqnarray}
since $\sum_{x\in
\mathbb{Z}_q} {\xi}^{-\hat{x}\hat{t}} = \left\{\begin{array}{ll} q, & \text{if} ~ t=0, \\
0, & \text{otherwise}.\end{array}\right.$
Then it follows that
\[\left|\cN_f(\uu)\right| = \left|\prod_{j=1}^n \omega^{-\hat{u}_j^2}\right| =1~.\]
Since $\uu \in \mathbb{Z}_q^n$ is arbitrary, the result follows.
\end{proof}

In the second construction, $2q$-negabent functions are constructed for $n=2$ and any positive integer $q \ge 2$.

\begin{theorem}\label{thm8}
Let $q \ge 2$ be any integer and $n=2$. Then the function $f: \mathbb{Z}_q^n \rightarrow \mathbb{Z}_{2q}$ defined by
\[f(x_1, x_2) = 2x_1x_2 + x_1\]
is a $2q$-negabent function.
\end{theorem}

\begin{proof}
For any $\uu = (u_1, u_2) \in \mathbb{Z}_q^2$, we have
\begin{align*}
&q\cN_f(\uu)  \\
& =  \sum_{(x_1, x_2) \in \mathbb{Z}_q^2}
\omega^{2\hat{x}_1\hat{x}_2+ \hat{x}_1} \xi^{\langle(\hat{u}_1, \hat{u}_2), (\hat{x_1}, \hat{x_2})\rangle} \omega^{\hat{x}_1+\hat{x_2}} \\
& =  \sum_{(x_1, x_2) \in \mathbb{Z}_q^2} \omega^{2(\hat{x}_2 + \hat{u}_1+1)\hat{x}_1 + (2\hat{u}_2+1)\hat{x}_2}  \\
& =  \sum_{x_1 \in \mathbb{Z}_q}\sum_{x_2 \in \mathbb{Z}_q} \omega^{2(\hat{x}_2 + \hat{u}_1+1)\hat{x}_1 + (2\hat{u}_2+1)\hat{x}_2} \\
& =   \sum_{x_2 \in \mathbb{Z}_q} \omega^{(2\hat{u}_2+1)\hat{x}_2} \sum_{x_1 \in \mathbb{Z}_q} \omega^{2(\hat{x}_2+\hat{u}_1+1)\hat{x}_1}\\
& =   \sum_{x_2 \in \mathbb{Z}_q} \omega^{(2\hat{u}_2+1)\hat{x}_2} \sum_{x_1 \in \mathbb{Z}_q} \xi^{(\hat{x}_2+\hat{u}_1+1)\hat{x}_1}~.
\end{align*}
Now
\[\sum_{x_1 \in \mathbb{Z}_q} \xi^{(\hat{x}_2+\hat{u}_1+1)\hat{x}_1}  = \left\{\begin{array}{ll} q, & \text{if} ~\hat{x}_2+\hat{u}_1+1=q, \\
0, & \text{otherwise}.\end{array}\right.
\]
Therefore,
\[ q\cN_f(\uu) = \omega^{(2\hat{u}_2+1)(q-\hat{u}_1-1)}q~,\]
which implies that $|\cN_f(\uu)| = 1$. Hence $f$ is a $2q$-negabent function. \end{proof}

\section{Examples} In this section,  we present
 some examples of  $2q$-negabent functions for different values of $q$ and $n$.

\begin{example}
The following two examples illustrate the construction of $2q$-negabent functions given in Theorem \ref{thm7}.
\begin{enumerate}
\item[(i)] Let $q=4$, $n=3$ and $f(x)= \hat{x}_1^2 + \hat{x}_2^2 + \hat{x}_3^2 - \hat{x}_1 - \hat{x}_2 - \hat{x}_3$. Then, after computation, we get $|\cN_f(\uu)| = \frac{1}{4^{3/2}} \times 8 = 1$ for all $\uu \in \mathbb{Z}_4^3$. Hence $f$ is a $2q$-negabent function for $q=4$.

\item[(ii)] Let $q=6$, $n=4$ and $f(x)= \hat{x}_1^2 + \hat{x}_2^2 + \hat{x}_3^2 + \hat{x}_4^2 - \hat{x}_1 - \hat{x}_2 - \hat{x}_3 - \hat{x}_4$. Then, after computation, we get $|\cN_f(\uu)| = \frac{1}{6^{4/2}} \times 36 = 1$ for all $\uu \in \mathbb{Z}_6^4$. Hence $f$ is a $2q$-negabent function for $q=6$.
\end{enumerate}
\end{example}

\begin{example}
This example illustrates the construction given in Theorem \ref{thm8}.

Let $n=2$, $q=3$, and $f\in\cNB_{2, 3}$ such that $f(x_1,
x_2)=2\hat{x}_1\hat{x}_2+\hat{x}_1$. Then the $2q$-NHT of $f$ at any $(u_1,
u_2)\in\BBZ_3^2$ is
\begin{align*}\label{eq-negabent-exp1}
&\cN_f(u_1, u_2) \\
& =  \frac{1}{3^{2/2}} \sum_{(x_1, x_2)\in\BBZ_3^2} \omega^{2\hat{x}_1\hat{x}_2+\hat{x}_1} \xi^{\langle(\hat{u}_1, \hat{u}_2), (\hat{x}_1, \hat{x}_2)\rangle}
\omega^{\hat{x}_1+\hat{x}_2}\\
& =  \frac{1}{3} \sum_{(x_1, x_2)\in\BBZ_3^2} \omega^{2\hat{x}_1\hat{x}_2+\hat{x}_1} \omega^{2(\hat{u}_1\hat{x}_1 + \hat{u}_2\hat{x}_2)}\omega^{\hat{x}_1+\hat{x}_2}\\
& =  \frac{1}{3}\sum_{x_1, x_2 \in \mathbb{Z}_3}\omega^{2(\hat{x}_2 + \hat{u}_1+1)\hat{x}_1 + (2\hat{u}_2+1)\hat{x}_2} \\
& =   \frac{1}{3} \sum_{x_2 \in \mathbb{Z}_3} \omega^{(2\hat{u}_2+1)\hat{x}_2} \sum_{x_1 \in \mathbb{Z}_3} \omega^{2(\hat{x}_2+\hat{u}_1+1)\hat{x}_1}\\
& =  \frac{1}{3}\left(3\omega^{(2\hat{u}_2+1)(2-\hat{u}_1)}\right)\\
& =  \omega^{(2\hat{u}_2+1)(2-\hat{u}_1)}~.
\end{align*}
Therefore, $|\cN_f(u_1, u_2)| = |\omega^{(2\hat{u}_2+1)(2-\hat{u}_1)}| = 1$, and hence $f$ is a $2q$-negabent function.
\end{example}

\begin{example}
Given below are some more  examples of $2q$-negabent functions.
\begin{enumerate}
    \item $f(x_1, x_2, x_3) = \hat{x}_1^2+\hat{x}_2^2+\hat{x}_3^2$ is a $2q$-negabent function on $3$ variables for $q$ an odd integer.

    \item $f(x_1, x_2, x_3, x_4) = \hat{x}_1^2+\hat{x}_2^2+\hat{x}_3^2+\hat{x}_4^2$ is a $2q$-negabent function on $4$ variables for  $q$ an odd integer.

    \item $f(x_1, x_2, x_3, x_4) = \hat{x}_1^2+\hat{x}_2^2+\hat{x}_3^2+\hat{x}_4^2 +2\hat{x}_1\hat{x}_2+2\hat{x}_3\hat{x}_4+2\hat{x}_2\hat{x}_4$
    is a $2q$-negabent function on $4$ variables for  $q=2, 3, 5, 7,9, \ldots. $

    \item $f(x)=\hat{x}^2+\hat{x}$ is a $2q$-negabent function on one
    variable for $q=2, 4, 6, 8, \ldots. $

    \item $f(x)=2\hat{x}^2+\hat{x}$ is a $2q$-negabent function on one
    variable for $q=3, 5, 7,  \ldots. $

    \item $f(x)=2\hat{x}^4+\hat{x}^2$ is a $2q$-negabent function on one
    variable for $q=9, 27, 81,  \ldots. $

    \item $f(x)=2\hat{x}^4+2\hat{x}^3+2\hat{x}^2+\hat{x}$ is a $2q$-negabent function on one
    variable  for $q=3, 4, 12. $

    \item $f(x_1, x_2)=\hat{x}_1^3+2\hat{x}_1\hat{x}_2+2\hat{x}_2^2$  is a $2q$-negabent
     function on two
    variables for $q=2$ and $3$.

    \item $f(x_1, x_2)=\hat{x}_1^3+2\hat{x}_1\hat{x}_2+\hat{x}_2^2$  is a $2q$-negabent
     function on two
    variables for $q=2, 3, 9, 27, 81, \ldots.$

    \item $f(x_1, x_2)= 2\hat{x}_1\hat{x}_2^2+2\hat{x}_1^2\hat{x}_2+2\hat{x}_1^2+
    2\hat{x}_2^2+\hat{x}_1+ \hat{x}_2$
     is a $2q$-negabent function on two
    variables for $q= 4. $
\end{enumerate}
\end{example}


\section{Conclusion}
In this paper we have introduced a new class of negabent functions. These functions are defined from $\mathbb{Z}_q^n$ to $\mathbb{Z}_{2q}$. We call these functions $2q$-negabent and the corresponding Hadamard transform as $2q$-nega-Hadamard transform. Some properties of $2q$-nega-Hadamard transform and those of $2q$-negabent functions are presented. Two constructions for $2q$-negabent functions are also presented. Some examples of these functions are given.

\end{document}